\documentclass[letterpaper, 10 pt, conference]{ieeeconf}  

\IEEEoverridecommandlockouts                              

\usepackage{amsmath} 
\usepackage{amssymb}  
\usepackage{mathtools}

\usepackage{xcolor}
\newtheorem{lemma}{\bf Lemma}
\newtheorem{theorem}{\bf Theorem}
\newtheorem{definition}{\bf Definition}
\newtheorem{assumption}{\bf Assumption}

\usepackage{acronym}
\acrodef{LTI}{linear-time-invariant}
\acrodef{LFT}{linear fractional transformation}
\acrodef{SDP}{semidefinite program}
\acrodef{QMI}{quadratic matrix inequality}

\usepackage{tikz}
\usepackage{pgfplots}
\pgfplotsset{compat=1.18}

\usepackage{cite}

\title{\LARGE \bf
A System Parametrization for Direct Data-Driven Analysis and Control with Error-in-Variables
}

\author{Felix Brändle and Frank Allgöwer
\thanks{F. Allg\"ower is thankful that this work was funded by the Deutsche Forschungsgemeinschaft (DFG, German Research Foundation) under Germany’s Excellence Strategy -- EXC 2075 -- 390740016 and within grant AL
	316/15-1 – 468094890. F. Br\"andle thanks the International Max Planck Research School for Intelligent Systems (IMPRS-IS) for supporting him.
	}
\thanks{F. Br\"andle and F. Allg\"ower are with the University of Stuttgart, Institute for Systems
		Theory and Automatic Control, 70550 Stuttgart,
		Germany. (e-mail: \{felix.braendle, frank.allgower\}@ist.uni-stuttgart.de)}%
}

\begin{document}
	
\maketitle
\thispagestyle{empty}
\pagestyle{empty}

\begin{abstract}
	In this paper, we present a new parametrization to perform direct data-driven analysis and controller synthesis for the error-in-variables case.
	To achieve this, we employ the Sherman-Morrison-Woodbury formula to transform the problem into a \acl{LFT} with unknown measurement errors and disturbances as uncertainties.
	For bounded uncertainties, we apply robust control techniques to derive a guaranteed upper bound on the $\mathcal{H}_2$-norm of the unknown true system.
	To this end, a single \acl{SDP} needs to be solved, with complexity that is independent of the amount of data.
	Furthermore, we exploit the signal-to-noise ratio to provide a data-dependent condition, that characterizes whether the proposed parametrization can be employed.
	The modular formulation allows to extend this framework to controller synthesis with different performance criteria, input-output settings, and various system properties.
	Finally, we validate the proposed approach through a numerical example.
\end{abstract}

\section{INTRODUCTION}
In recent years, interest in system analysis and controller design based on collected data has been continuously increasing \cite{Coulson2019, Berberich2020, Martin2023a}.
Indirect data-driven control first identifies a model to employ model-based techniques for controller design and analysis as a second step.
In practice, only a finite number of possibly noisy data points are available, making identifying the true system challenging in many cases, even for the linear case \cite{Oymak2019}. 
This mismatch between the true and identified model may lead to instabilities or deteriorated performance guarantees when applying the controller to the true system.

A possible remedy is offered by direct data-driven control, where end-to-end guarantees can be provided even for a finite number of noisy data samples \cite{Romer2019}.
Here, available results are often restricted to bounded noise, which includes, e.g., energy bounds and individual noise bounds for each time instance.
Many results in the literature build on the data-informativity framework \cite{Waarde2020a, Waarde2022}, which provides stability and performance guarantees for the true system by verifying a desired property for all systems consistent with the data.
Another direct data-driven approach relies on parameterizing the set of consistent system matrices in terms of a right inverse of the stacked data matrices \cite{DePersis2020, Berberich2019}.

Most existing works consider the case of additive process noise \cite{Waarde2020a, Koch2020}.
Although additive process noise often appears in linear dynamical systems, it is possible to approximate nonlinear systems via basis functions and treat the resulting approximation error as an additive disturbance \cite{Martin2023a}. 
By deriving a bound on the approximation error, it is possible to provide guarantees for unknown nonlinear systems.

When collecting data, the occurring measurement noise introduces yet another error source leading to the error-in-variables problem.
Without explicitly addressing the measurement noise, the derived analysis often results in biases. 
For example, in the linear least squares problem, not considering error-in-variables results in a systematic underestimation of the absolute parameter values \cite{Soederstroem2003}.
To address this the total least squares problem must be solved.

In system identification, the bias can be eliminated based on set-membership estimation, where the set of consistent systems is identified \cite{Cerone2018}. 
Similarly, \cite{Miller2023c} designs a stabilizing controller using such an identified set. 
To achieve this, the authors reformulate the problem as a sum-of-squares problem, which can be solved iteratively using \acp{SDP}.
To design guaranteed stabilizing controllers for sufficiently small measurement noise, an alternative is to introduce a regularization term to a predictive controller based on Hankel matrices \cite{Berberich2022b}. 
This leads to a trade-off between a simple parametrization and reliance on data.
Moreover, \cite{DePersis2020} treats error-in-variables as additive process noise, which requires additional knowledge on the true system. 
However, this knowledge may not always be available.
Recent results extend the data-informativity setting to the error-in-variables case under the assumption that the process noise and the measurement error satisfy a particular quadratic matrix inequality \cite{Bisoffi2024} .

In this work, we derive an explicit and exact parametrization of all systems consistent with the data by interpreting the data-driven control problem as linear regression.
The main contribution of this paper is to generalize the results of \cite{DePersis2020, Koch2020} to the errors-in-variables problem. 
Moreover the proposed parameterization allows for a more flexible description of the noise affecting the data than the data-informativity framework, while still requiring the solution of only a single \ac{SDP}.
Furthermore, we exemplify the use of the derived parametrization to compute an upper bound on the $\mathcal{H}_2$-norm of a unknown system.
Finally, by formulating the parameterization in terms of a \ac{LFT}, many already existing methods from robust control can be applied to extend this setup for different performance criteria, noise descriptions and controller design.

\section{PRELIMINARIES}

We denote the $n\times n$ identity matrix and the $p\times q$ zero matrix as $I_n$ and $0_{p\times q}$, respectively.
We omit the indices if the dimensions are clear from the context. 
For the sake of space limitations, we use $[\star]$ if the corresponding matrix block can be inferred from symmetry. 
Moreover, we use $P\succ 0$ and $P\succeq 0$ if the symmetric matrix $P$ is positive definite or positive semidefinite. 
Negative (semi-)definiteness is defined by $P\prec 0$ and $P\preceq 0$, respectively. We use $\mathrm{diag}(V_1,V_2)$ to describe a block diagonal matrix with $V_1$ and $V_2$ on its diagonal.
Moreover, we use $\sigma_\mathrm{min}(A)$ and $\sigma_\mathrm{max}(A)$ to denote the minimal and maximal singular value of $A$.
For measurements $\{x_k\}_{k=0}^{N}$, we define the following matrices
\begin{equation*}
	X=\begin{bmatrix}
		x_0 & x_1 & \ldots & x_{N-1}
	\end{bmatrix}, 
	\text{ } X_{+}=\begin{bmatrix}
		x_1 & x_2 & \ldots & x_{N} 
	\end{bmatrix}.  		
\end{equation*}

In this paper, we consider matrices with unknown elements. In particular, we look at \acp{LFT}.
\begin{definition}
	The \acl{LFT}
	\begin{align}
		\begin{bmatrix}
			y \\ z
		\end{bmatrix} &= 
		\left[\begin{array}{cc}
			A & B \\
			C & D \\
		\end{array}\right]
		\begin{bmatrix}
			x \\ w 
		\end{bmatrix}, & w &= \delta z
	\end{align}
	is well-posed if the matrix $(I-D\delta)$ is non-singular. 
\end{definition}

In this work, we treat $\delta$ as an unknown matrix. 
Furthermore, a well-posed \ac{LFT} can equivalently be written as
\begin{equation}
	y = 	\left(A + B \delta \left(I-D\delta\right)^{-1}C\right) x. \label{eq:Pre:LFT}
\end{equation}

The next lemma is key in deriving an \ac{LFT} formulation in Section\,\ref{sec:Param}.
\begin{lemma} \label{lemma:param:Woodbury}
	Let $I+UV$ be invertible. Then the following equality holds
	\begin{equation}
		(I+UV)^{-1} = I-U(I+VU)^{-1}V.
	\end{equation}
\end{lemma}

This is a special case of the well-known Sherman–Morrison–Woodbury formula \cite{Woodbury1950}.

\section{DATA-DRIVEN SYSTEM PARAMETRIZATION}\label{sec:Param}

In this work, we consider a linear regression model of the form
\begin{equation}
	y = \Theta_{\mathrm{tr}}x,
\end{equation}
with regressand $y\in\mathbb{R}^p$, regressor $x\in\mathbb{R}^n$ and unknown true parameter matrix $\Theta_{\mathrm{tr}}\in\mathbb{R}^{p \times n}$.
Our goal is to find an \ac{LFT} parametrization of the true parameter matrix  $\Theta_{\mathrm{tr}}$ based on a finite number of measurements. 
To this end, we collect data $\{x_k, y_k\}_{k=0}^{N-1}$, where
\begin{align*}
	  x_k = x_{\mathrm{tr},k} + \tilde{v}_{1,k}, \quad y_k =  y_{\mathrm{tr},k} + \tilde{v}_{2,k}, \quad k=0,\ldots, N-1,
\end{align*}
with $x_{\mathrm{tr},k}$ and $y_{\mathrm{tr},k}$ being the true regressand and regressor.
For simplification, we rewrite the problem in matrix notation
\begin{equation}
	Y-\tilde{V}_2 = \Theta_{\mathrm{tr}}(X-\tilde{V}_1)	\label{eq:Param:GeneralRegr}
\end{equation}
with $Y\in\mathbb{R}^{p \times N}$ and $X\in\mathbb{R}^{n \times N}$ being the horizontally stacked regressor and regressand. To distinguish between the error sources, we denote the error term in the regressor $\tilde{V}_1 = L_1 V_1 R_1$ with measurement error and the error term in the regressand $\tilde{V}_2 = L_2 V_2 R_2$ by disturbance. The matrices $L_1$, $R_1$, $L_2$ and $R_2$ are introduced for a more flexible parametrization. For example to account for measurement errors in  only certain elements of the regressor one must set the corresponding rows of $L_1$ to zero. In contrast, methods using the data-informativity framework require an assumption on $\tilde{V}_1$ and $\tilde{V}_2$ directly \cite{Waarde2020a, Bisoffi2024}, hence this setup provides additional flexibility.

Using the collected data, we define the set $\Delta$ of all measurement errors and disturbances consistent with the data
\begin{align*}
	\Delta = \{(V_1,V_2)\in \Delta_0 \,|\,&\exists \Theta\in\mathbb{R}^{p \times n}\colon \\ 
	&Y-L_2V_2R_2  = \Theta(X-L_1V_1R_1) \},
\end{align*}
where $\Delta_0$ incorporates prior knowledge of the errors, e.g., bounds on the spectral norm \cite{Berberich2019}. 
The matrix $\Theta$ takes the role of the unknown parameter matrix.
Moreover, we define the set $\Sigma_{\Theta}$ of all parameter matrices consistent with the data as
\begin{align*}
	\Sigma_{\Theta} = \{\Theta\in\mathbb{R}^{p \times n}\,|\,&\exists\,(V_1,V_2)\in \Delta\colon\\
	&Y-L_2V_2R_2  = \Theta(X-L_1V_1R_1)\}.
\end{align*}
Furthermore, we assume the data satisfy the following Assumption.
\begin{assumption} \label{ass:Param:G}
	There exists a matrix $G\in\mathbb{R}^{N\times n}$ satisfying
	\begin{align}
		XG&=I,		\label{eq:Param:GRightInverse}\\
		\mathrm{rank}\left(I-L_1V_1R_1 G\right) &= n, \quad\forall(V_1,V_2)\in \Delta. \label{eq:Param:XTrRightInverse}
	\end{align}
\end{assumption}

Condition\,\eqref{eq:Param:GRightInverse} implies that $G$ is a right inverse of the regressor $X$, which requires $X$ to have full row rank. 
This condition can be checked using the data and imposes a condition on the data collection to contain sufficient information.
Condition\,\eqref{eq:Param:XTrRightInverse} is more difficult to validate since it depends on the unknown measurement error. 
To this end, we use \eqref{eq:Param:GRightInverse} to rewrite \eqref{eq:Param:XTrRightInverse} as
\begin{align*}
	I-L_1V_1R_1 G = (X-L_1V_1R_1)G.
\end{align*}
Condition \eqref{eq:Param:XTrRightInverse} is satisfied if $\sigma_\mathrm{min}(X)\geq\sigma_\mathrm{max}(L_1V_1R_1)$, i.e. $V_1$ is sufficiently small in comparison to $X$.
Consequently, one way to satisfy \eqref{eq:Param:XTrRightInverse} is via collecting data with a sufficiently large signal-to-noise-ratio. Similar assumptions can be found in \cite[Assumption 1]{Miller2023c} and \cite[Assumption 1]{Bisoffi2024} to ensure a sufficient signal-to-noise-ratio and hence compactness of the corresponding sets.

To find an explicit parametrization of $\Sigma_{\Theta}$, we define the following set
\begin{align*}
	\Sigma_{\Theta}^{G} = \{&\Theta\in\mathbb{R}^{p \times n}\,|\,\exists(V_1,V_2)\in \Delta\colon \\
	&\Theta = (Y-L_2V_2R_2) G\left(I-L_1V_1R_1 G\right)^{-1}\}, 
\end{align*}
for some matrix $G\in\mathbb{R}^{N\times n}$.

\begin{lemma} \label{lemma:param:SetEquality}
	Suppose $G\in\mathbb{R}^{N\times n}$ satisfies Assumption\,\ref{ass:Param:G}. Then, $\Sigma_{\Theta} = \Sigma_{\Theta}^{G}$.
\end{lemma}
\begin{proof}
	 In the following, we show the equivalence of both parametrizations by considering each set inclusion separately.
	
	$\Sigma_{\Theta}\subseteq\Sigma_{\Theta}^{G}$:
	Take any $\Theta\in\Sigma_{\Theta}$ and let $(V_1,V_2)\in \Delta$ be the corresponding error terms. Hence, we directly deduce
	\begin{equation}
		Y-L_2V_2R_2 = \Theta(X-L_1V_1R_1).
	\end{equation}
    Post-Multiplication with $G$ and applying \eqref{eq:Param:GRightInverse} yield
	\begin{equation}
		(Y-L_2V_2R_2)G = \Theta \left(I-L_1V_1R_1 G\right).
	\end{equation}
	Then, we exploit \eqref{eq:Param:XTrRightInverse} to deduce invertibility of $(I-L_1V_1R_1G)$, hence $\Sigma_\Theta\subseteq\Sigma_{\Theta}^{G}$. 
	
	$\Sigma_{\Theta}^{G}\subseteq\Sigma_{\Theta}$:
	Let $(V_1,V_2)\in \Delta$ be an error and $\theta$ be the corresponding parameter matrix. We need to show that
	\begin{equation}
		\Theta \left(I-L_1V_1R_1 G\right) = (Y-L_2V_2R_2)G
	\end{equation}
	holds, while 
	\begin{equation}
		\Theta (X-L_1V_1R_1) = (Y-L_2V_2R_2).
	\end{equation}
	is also satisfied. Combining both equations results in
	\begin{equation*}
		\Theta \begin{bmatrix}
			X-L_1V_1R_1 &
			I-L_1V_1R_1 G
		\end{bmatrix} = (Y-L_2V_2R_2)\begin{bmatrix} I & G\end{bmatrix}.
	\end{equation*}
	Using \eqref{eq:Param:GRightInverse}, this is equivalent to
	\begin{equation}
		\Theta (X-L_1V_1R_1)
		\begin{bmatrix}	
			I & G
		\end{bmatrix} = (Y-L_2V_2R_2)\begin{bmatrix} I & G\end{bmatrix}.
	\end{equation}
	By definition of $\Delta$, this has a solution, ensuring $\Theta$ has the desired structure and is consistent with the data. Hence $\Sigma_\Theta^G\subseteq\Sigma_{\Theta}$.
\end{proof}

Lemma\,\ref{lemma:param:SetEquality} provides an exact parametrization of the uncertain model. The main advantage of using $\Sigma_{\Theta}^{G}$ is that it offers an explicit description, while $\Sigma_{\Theta}$ is an implicit description. 
A similar parameterization in terms of a right inverse $G$ is presented in \cite{Berberich2019,DePersis2020,Koch2020} without measurement error.

In the next step, we reformulate the explicit parametrization $\Sigma_{\Theta}^{G}$ as an \ac{LFT} to describe $\Sigma_{\Theta}$.
\begin{theorem} \label{theorem:param:LFT}
	Suppose $G\in\mathbb{R}^{N\times n}$ satisfies Assumption\,\ref{ass:Param:G}. Then, $\Sigma_{\Theta}$ can be described by the well-posed \ac{LFT}
	\begin{align}
		\left[\begin{array}{c}
			y \\ \hline z_{1} \\ z_{2}
		\end{array}\right]
		&=
		\left[ \begin{array}{c|cc}
			YG &    YG L_1 & -L_2\\ \hline
		  R_1G &  R_1G L_1 &    0\\  
		  R_2G &  R_2G L_1 &	0 
		\end{array}\right]
		\left[\begin{array}{c}
			x \\ \hline w_{1} \\ w_{2}
		\end{array}\right], \label{eq:param:FullLFT}\\ 
		\begin{bmatrix}
			w_{1} \\
			w_{2}
		\end{bmatrix} &= \mathrm{diag}(V_1,V_2)\begin{bmatrix}z_{1} \\ z_{2}\end{bmatrix},\qquad
		 	(V_1,V_2)\in \Delta.
	\end{align}
\end{theorem}

\begin{proof}
	First, we apply Lemma \ref{lemma:param:SetEquality} to get an equivalent explicit description of $\Sigma_\Theta$. Next, using \eqref{eq:Param:XTrRightInverse} and Lemma\,\ref{lemma:param:Woodbury}, we arrive at
	\begin{equation*} 
		\Theta = (Y-L_2V_2R_2) G
		\left( I + L_1V_1\left(I-R_1GL_1 V_1\right)^{-1}R_1G\right).
	\end{equation*}
	Furthermore, we define
	\begin{align}
		&z_1 = \left(I-R_1GL_1 V_1\right)^{-1}R_1G x, & w_1 = V_1 z_1,
	\end{align}
	and compare with \eqref{eq:Pre:LFT} to get
	\begin{equation*}
		\begin{bmatrix}
			y \\  z_{1}
		\end{bmatrix}
		=
		\left[ \begin{array}{cc}
			(Y-L_2V_2R_2)G & (Y-L_2V_2R_2)GL_1\\ 
			R_1 G &               R_1 G L_1\\  
		\end{array}\right]
		\begin{bmatrix}
			x \\ w_{1}
		\end{bmatrix}.
	\end{equation*}
	For
	\begin{align}
		&z_{2} = R_2G x + R_2GL_1w_{1}, &
		w_{2} = V_2 z_{2},
	\end{align}
	we receive \eqref{eq:param:FullLFT}. Well-posedness follows directly from
	\begin{equation*}
		I - \begin{bmatrix} R_1GL_1 & 0\\ R_2GL_1 & 0	\end{bmatrix} \begin{bmatrix}
			V_1 & 0 \\
			0 & V_2
		\end{bmatrix} = \begin{bmatrix}
			I-R_1GL_1V_1 & 0 \\
			-R_2GL_1V_1 & I
		\end{bmatrix}
	\end{equation*}
	and \eqref{eq:Param:XTrRightInverse} with Lemma\,\ref{lemma:param:Woodbury} implying invertibility.
\end{proof}

Theorem\,\ref{theorem:param:LFT} provides an exact parametrization of all para\-meter matrices consistent with the data and $G$ as a free decision variable. Furthermore, it offers an interesting relation to system identification methods. 
In particular, for the nominal case without any disturbance or measurement error $	Y = \Theta_{\mathrm{tr}}X$, it is possible to reconstruct $\Theta_{\mathrm{tr}}=YG$ by taking any $G$ satisfying \eqref{eq:Param:GRightInverse}. 
Thus, \eqref{eq:param:FullLFT} can be interpreted as the identified nominal model for $V_1$ and $V_2$ being zero. 
Hence, for small errors, we recover the nominal case.

Furthermore, there is a clear difference how $V_1$ and $V_2$ affect the system. Similar to the results in \cite{Berberich2019} and  \cite{Koch2020}, the \ac{LFT} in Theorem\,\ref{theorem:param:LFT} includes an explicit description of the uncertainty channel for $V_2$, since $z_2$ is independent of $w_2$.
Meanwhile, the channel $z_1$ depends on $w_1$, resulting in an implicit description to describe the inverse. 
Future research may also investigate the term $Gx+GL_1w_{1}$, which appears in both the description of $y$ and the uncertainty channel, possibly leading to meaningful insights into the error-in-variables case.

\section{$\mathcal{H}_2$-Analysis} \label{sec:Analysis}

Now, we demonstrate the usage of Theorem\,\ref{theorem:param:LFT} to derive an upper bound for the $\mathcal{H}_2$-norm of an unknown discrete-time \ac{LTI} system. 
\begin{definition}
	The $\mathcal{H}_2$-norm of an asymptotically stable, discrete-time transfer function $P$ is
	\begin{equation}
		||P(e^{\mathrm{j}\omega}))||_2 = \sqrt{\frac{1}{2\pi} \int\limits_{-\pi}^{\pi} \mathrm{trace}(P(e^{\mathrm{j}\omega})^*P(e^{\mathrm{j}\omega}))\,\mathrm{d}\omega}
	\end{equation}
	with $^*$ denoting the conjugate transpose.
\end{definition}

The squared $\mathcal{H}_2$-norm coincides with the total output-energy of the impulse response \cite{Scherer2000}.
We employ convex relaxations and multiplier techniques from robust control to compute an upper bound of the $\mathcal{H}_2$-norm.
To this end, we derive an upper bound for all systems within $\Sigma_{\Theta}^G$, which also establishes a guaranteed upper bound for the true system.

Stating the parametrization as an \ac{LFT} is advantageous, because it allows to use a plethora of already existing methods from robust control to solve different control problems. In this paper, we focus on the $\mathcal{H}_2$-norm, but to design controllers only the control input and output have to be included in the regressor and regressand, see \cite{Weiland1994} for more details. Nonlinear control can be addressed by including nonlinear basis functions in the regressor \cite{Strasser2021}. Despite focusing on \acp{QMI} in this work, different uncertainty descriptions like integral-quadratic-constraints or polytopes are also possible.

In this section, we consider an \ac{LTI} system of the form
\begin{equation}
	\left[ \begin{array}{c}
		x_{k+1} \\ \hline
		z_{p,k}
	\end{array} \right] = 
	\left[ \begin{array}{c|cc}
		A_{\mathrm{tr}}   & B_{p,\mathrm{tr}}  & B_{\tilde{d}} \\ \hline
		C_{p,\mathrm{tr}}   & D_{p,\mathrm{tr}}  & 0
	\end{array} \right] 
	\left[ \begin{array}{c}
		x_{k} \\ \hline
		w_{p,k} \\
		\tilde{d}_{k}
	\end{array} \right], \label{eq:Anal:LTIPerformance}
\end{equation}
where $x_k\in\mathbb{R}^n$ is the state vector, $w_{p,k}\in\mathbb{R}^{m_p}$ is the performance input, $\tilde{d}_k\in\mathbb{R}^{m_d}$ is an external process disturbance, and $z_{p,k}\in\mathbb{R}^{p_p}$ is the performance output at time step $k\geq 0$. We aim to find an upper bound on the $\mathcal{H}_2$-norm from $w_p$ to $z_p$.

The true system is described by the unknown matrices $A_{\mathrm{tr}}$, $B_{p,\mathrm{tr}}$, $C_{p,\mathrm{tr}}$, $D_{p,\mathrm{tr}}$ and the known matrix $B_{\tilde{d}}$, which is used to include additional prior knowledge on the influence of the process disturbance $\tilde{d}_k$. Moreover, we assume measurement errors in the obtained data samples $x_k$ and $z_{p,k}$. 
To illustrate the benefit of including $L_1$ and $R_1$, we select $w_{p}$ to excite the system.
As we can choose the corresponding signal, we know the true value of $w_{p}$.
This known structure on the measurement noise can be included by setting the last $m_p$ rows of $L_1$ to $0$.
In contrast, \cite{Bisoffi2024} requires that $w_{p}$ is not known exactly. The matrix notation of the system dynamics reads
\begin{equation}
	\left[\begin{array}{c} X_+ - \tilde{V}_{X_+} \\ Z_p - \tilde{V}_{Z_p}\end{array}\right] = 	
	\Theta_{\mathrm{tr}}\left[\begin{array}{c}	X-\tilde{V}_X \\ W_p \end{array}\right]	
	+\left[\begin{array}{c} B_{\tilde{d}} \tilde{D} \\ 0\end{array}\right] ,
\end{equation}
where $\Theta_{\mathrm{tr}}$ contains the stacked unknown, true system matrices. 
Interestingly, the measurement errors in $X_+$ and the process disturbance $\tilde{D}$ can be treated the same way, since we only consider the next time step. 

Since for any combination of $Y$, $X$ and $\theta$, it is possible to explain the collected data with appropriate measurement errors and disturbances, we impose bounds on the error terms.
We assume that $\tilde{V}_X=\tilde{L}_{V_X}V_X\tilde{R}_{V_X}$ is bounded and can be described by a known \ac{QMI} \cite{Berberich2019}
\begin{align*}
	V_X &\in \left\{V_X\mid\star^\top \left[\begin{array}{cc}
	Q_{V_X} & 0 \\ 0 & R_{V_X}\end{array}\right]
	\left[\begin{array}{c}
		V_X-V_{X,0} \\ I 
	\end{array}\right] \succeq 0 \right\}
\end{align*}
with $R_{V_X}\succ 0$ and $Q_{V_X}\prec 0$. 
For simplicity, we consider only the case $V_{X,0}=0$, otherwise $V_{X,0}$ can be moved in the regressor or regressand.
Similar characterizations are made for $\tilde{D}$, $\tilde{V}_{X_+}$ and $\tilde{V}_{Z_p}$. 
Now, we rewrite the system as in \eqref{eq:Param:GeneralRegr} with $[X^\top,W_p^\top]^\top$ and $[X_+^\top,Z_p^\top]^\top$ taking the role of regressor and regressand, $V_1=V_X$ and $V_2=\mathrm{diag}(V_{X_+},V_{Z_p},D)$. 
Note that repeated entries in $V_X$ and $V_{X_+}$ can be represented using extended multiplier classes, as in \cite{Berberich2020}.
For space reasons, we omit these extended multipliers and consider $V_X$ and $V_{X_+}$ to be independent.
Furthermore, we have
\begin{align*}
	L_1 &= \left[\begin{array}{c} \tilde{L}_{V_X} \\ 0	\end{array}\right],&
	L_2 = \left[\begin{array}{ccc} \tilde{L}_{V_{X_+}} & 0 & B_d\tilde{L}_D\\ 0 & \tilde{L}_{V_{Z_p}} & 0\end{array}\right], \\
	R_1 &= \tilde{R}_{V_X},&
	R_2 = \left[\begin{array}{c}
		\tilde{R}_{V_{X_+}}\\ \tilde{R}_D\\ \tilde{R}_{Z_p}
	\end{array}\right]
\end{align*}
and
\begin{align*}
	&\Delta_0 = \left\{(V_1,V_2)\mid \left[\begin{array}{c}\star\\\star\end{array}\right]^\top P \left[\begin{array}{cc}\mathrm{diag}(V_1,V_2) \\ I\end{array}\right]\succeq0\right\}, \\
	&P=\mathrm{diag}(Q_{V_X},Q_{V_{X_+}},Q_{V_{Z_p}},Q_D,R_{V_X},R_{V_{X_+}},R_{V_{Z_p}},R_D).
\end{align*}
This formulation can include free decision variables in $P$ to reduce conservatism, i.e., scaling the \ac{QMI} of each error source by a positive scalar, since this won't change the uncertainty set.
Next, we discuss suitable convex relaxations to ensure the problem is solvable by standard robust control techniques. 
First, we overapproximate the set $\Delta$ by
\begin{equation}
	\tilde{\Delta}=\left\{\delta \mid \left[\begin{array}{c}\delta \\ I\end{array}\right]^\top P \left[\begin{array}{cc}\delta \\ I\end{array}\right]\succeq0\right\}.
\end{equation}
In particular, we relax the restriction that the set of errors must be consistent with the data as in \cite[Section IV]{Koch2020}, hence we only consider a superset.
This is done in order to employ robust control techniques, which requires an \ac{QMI} formulation.
Furthermore, we do not require the uncertainty to be block diagonal, which is common practice, when combining multiple uncertainties by stacking the corresponding multipliers diagonally \cite[Section 6.3]{Weiland1994}. 
An interesting topic for future research is how to choose $\tilde{\Delta}$ as close as possible to $\Delta$ to reduce the introduced conservatism, e.g., by including structured singular values to enforce the block diagonal structure.

Now, we can apply Theorem\,\ref{theorem:param:LFT} to obtain the following \ac{LFT}
\begin{align}\label{eq:Anal:LFT}
	\left[\begin{array}{c}
		x_{k+1} \\\hline z_{p,k} \\ z_k
	\end{array}\right] &= 
	\left[\begin{array}{c|cc}
		A &    B_1 &    B_2 \\\hline
		C_1 &  D_{1} & D_{12} \\
		C_2 & D_{21} &    D_2 
	\end{array}\right]
	\left[\begin{array}{c}
		x_k \\\hline w_{p,k} \\ w_k
	\end{array}\right], \\
	w_k &= \delta z_k, \qquad \delta\in\tilde{\Delta},  \nonumber
\end{align}
with system matrices depending on the regressand $Y$, and the right inverse $G$, such that the true system can be represented by a particular $\delta\in\tilde{\Delta}$.

As already mentioned before $G$ is a decision variable, which can be employed to reduce conservatism as demonstrated in \cite{DePersis2020} without error-in-variables.
Finding an optimal $G$ to minimize the $\mathcal{H}_2$-norm can be reformulated as a standard control problem by extracting $G$ from \eqref{eq:param:FullLFT} to get
\begin{align*}
	\left[\begin{array}{c}
		y \\ \hline z_{1} \\ z_{2} \\ \hline z
	\end{array}\right]
	&=
	\left[ \begin{array}{c|cc|c}
		0 &   0 & -L_2 & Y \\ \hline
		0 &   0 &    0 & R_1\\  
		0 &   0 &	 0 & R_2\\ \hline
		I & L_1 &    0 & 0
	\end{array}\right]
	\left[\begin{array}{c}
		x \\ \hline w_{1} \\ w_{2} \\ \hline u
	\end{array}\right] 
	&u = Gz.
\end{align*}
This corresponds to the static output feedback case with $G$ as feedback gain and is not convexifiable in general \cite{Sadabadi2016}.
The case for $L_1=0$, i.e., no error-in-variables leads to static state-feedback, which can be convexified as seen in \cite{Koch2020}.
In the following, we choose a fixed $G$ and refer to \cite{Sadabadi2016} to include an iterative scheme to optimize for $G$.
According to \cite[Section 4]{Braendle2024}, a suitable way to choose $G$ is
\begin{align}
	\hat{G} &= R^{-1}X^\top(XR^{-1}X^\top)^{-1}, \label{eq:Analysis:Gopt}\\
	R &= [\star]^\top \mathrm{diag}(R_{V_X},R_{V_{X_+}},R_{V_{Z_p}},R_D) \begin{pmatrix} R_1 \\ R_2\end{pmatrix} \nonumber
\end{align}
and $X$ being the regressor to minimize the volume of $\tilde{\Delta}$. Note that due to the additional shift $YG$, this may not coincide with the optimal $G$ to minimize the $\mathcal{H}_2$-norm and requires $[R_1^\top \, R_2^\top]^\top$ to have full column rank.
Another potential choice for $G$ is the Moore-Penrose inverse, due to being the right inverse with the smallest norm \cite{Penrose1956}.
This leads to small values of $YG$ and a small amplification of the uncertainty.
Next, we introduce a Theorem to compute an upper bound on the $\mathcal{H}_2$.
\begin{theorem} \label{theorem:Analysis:RobustH2}
	The \ac{LFT} \eqref{eq:Anal:LFT} is well-posed and $\gamma$ is a guaranteed upper bound on the $\mathcal{H}_2$-norm, if there exist $\mathcal{X}\succ0$, $\mathcal{Z}\succ0$, $P_1$, and $P_2$ such that
	\begin{equation}
		\left[\begin{array}{c}\delta \\ I \end{array}\right]^\top P_i \left[\begin{array}{c}\delta \\ I\end{array}\right] \succeq 0, \qquad \forall \delta\in\tilde{\Delta}, i=1,2,
	\end{equation}
	and 
	\begin{align*}
		\mathrm{tr}(\mathcal{Z}) < \gamma^2,
	\end{align*}
	\begin{align*}
		\left[\begin{array}{cc}
			\star & \star\\
			\star & \star\\ \hline
			\star & \star	\\
			\star & \star \\\hline
			\star & \star
		\end{array}\right]^\top	
		&\left[\begin{array}{cc|c|c}
			-\mathcal{X} &           0 &   0 &   0 \\
			0 & \mathcal{X} &   0 &   0 \\\hline
			0 &           0 & P_1 &   0 \\\hline            
			0 &           0 &   0 &   I \\
		\end{array}\right]
		\left[\begin{array}{cc}
			I & 0\\
			A & B_2\\\hline
			0 & I	\\
			C_2 & D_2 \\\hline
			C_1 & D_{12}
		\end{array}\right]
		\prec 0, \\
		\left[\begin{array}{cc}
			\star & \star\\
			\star & \star\\ \hline
			\star & \star	\\
			\star & \star \\\hline
			\star & \star
		\end{array}\right]^\top	
		&\left[\begin{array}{cc|c|c}
			-\mathcal{Z} &           0 &   0 &   0 \\
			0 & \mathcal{X} &   0 &   0 \\\hline
			0 &           0 & P_2 &   0 \\\hline            
			0 &           0 &   0 &   I \\
		\end{array}\right]
		\left[\begin{array}{cc}
			0 & I\\
			B_2 & B_1\\\hline
			I & 0	\\
			D_2 & D_{21} \\\hline
			D_{12} & D_{1}
		\end{array}\right]
		\prec 0.
	\end{align*}	
\end{theorem}

\begin{proof}
	The proof follows standard robust control arguments, compare, e.g. \cite[Theorem 10.6]{Scherer2000}.
\end{proof}
If $P_i$ $i=1,2$ are parameterized affinely in the decision variables, the conditions in Theorem \ref{theorem:Analysis:RobustH2} become an \ac{SDP}, which  can be solved by standard numerical solvers. 
Moreover, it is possible to directly minimize $\gamma^2$, while solving the \ac{SDP} to obtain a potentially small upper bound on the true $\mathcal{H}_2$-norm. 
One possibility to guarantee satisfaction of Assumption\,\ref{ass:Param:G} is to apply persistently exciting inputs with a sufficiently large signal-to-noise ratio \cite{Willems2005}. 
Note that Theorem \ref{theorem:Analysis:RobustH2} already guarantees well-posedness of the \ac{LFT}, hence providing a verifiable condition to check \eqref{eq:Param:XTrRightInverse}. 
Existence of $G$ can be validated by checking that the regressor has full rank.
Note that the size and number of decision variables of the \ac{SDP} grows with the dimension of states, performance inputs, and $w_k$.
To reduce the dimension of the uncertainty the following lemma can be used.
\begin{lemma}[\cite{Braendle2024}] \label{lemma:Ana:GTrick}
	Suppose $R\succ 0$ and $E$ is a matrix with full column rank, then
	\begin{align*}
		&\left\{ \delta E \mid [\star]^\top \begin{bmatrix} - Q & 0 \\ 0 & R	\end{bmatrix} \begin{bmatrix} \delta \\ I\end{bmatrix} \succeq 0 \right\} \\
		= &\left\{\tilde{\delta} \mid [\star]^\top \begin{bmatrix} - Q & 0 \\ 0 & E^\top R E	\end{bmatrix} \begin{bmatrix} \tilde{\delta} \\ I\end{bmatrix} \succeq 0\right\}.
	\end{align*}
\end{lemma}
Since the number of columns in $G$ is fixed by the state and input dimension, using Lemma\,\ref{lemma:Ana:GTrick} with $E=[R_1^\top\,R_2^\top]^\top G$ changes the uncertainty to have at most $n+m_p$ columns.
For the case of $\mathcal{H}_2$-norm analysis, this won't reduce the size of the \ac{SDP} itself, but for state-feedback synthesis, the size of the \ac{SDP} increases with the size of $\delta$ \cite{Berberich2019}.


\section{NUMERICAL EXAMPLE}\label{sec:Example}
In this section, we validate the results of Section \ref{sec:Analysis} with a numerical example. To this end, we consider the system as described in \eqref{eq:Anal:LTIPerformance} with the following system matrices
\begin{align*}
	A_{\mathrm{tr}} &= \begin{bmatrix}
		          1 &   0.2 &    0 &   0 \\
		         -1 &   0.5 &  0.6 & 0.3 \\
		          0 &     0 &    1 & 0.2 \\
		        0.3 &  0.15 & -0.3 & 0.85
	\end{bmatrix}, 
	B_{p,\mathrm{tr}} = \begin{bmatrix}
		0 & 0 \\
		0.2 & 0\\
		0 & 0\\
		0 & 0.1
	\end{bmatrix},\\		
	B_{\tilde{d}} &= \begin{bmatrix}	0 \\ 0 \\ 0 \\ 0.2	\end{bmatrix}, 
	C_{p,\mathrm{tr}} = \begin{bmatrix} 1 & 0 & 0 & 0\\ 0 & 0 & 1 & 0	\end{bmatrix},
	D_{p,\mathrm{tr}} = \begin{bmatrix} 0 & 0 \\ 0 & 0\end{bmatrix}.
\end{align*}
We collect data $\{x_k\}_{k=0}^{N-1}$ and $\{w_{p,k},z_{p,k}\}_{k=0}^{N-2}$ for different data lengths $N\in [n+m_{p}+1,300]$. For $N$ smaller than $n+m_{p}+1$, $[X^\top,W_p^\top]^\top$ will have more rows than columns. Thus Assumption\,\ref{ass:Param:G} is violated. To further illustrate the benefit of parameterizing $\tilde{V}_1=L_1V_1R_1$ and $\tilde{V}_2=L_2V_2R_2$, we assign  a random constant value $\tilde{d}\in[-\bar{d},\bar{d}]$ with $\bar{d}=0.01$ to the disturbance $\tilde{d}$, leading to the description:
\begin{align}
	\tilde{D}&= d\begin{bmatrix}1 & 1 & \cdots& 1	\end{bmatrix},  \\
	d&\in\left\{d\mid\begin{pmatrix}d \\ 1\end{pmatrix}^\top 
	\begin{pmatrix}	-1 & 0 \\ 0 & \bar{d}^2	\end{pmatrix}
	\begin{pmatrix} d \\1\end{pmatrix} 
	\geq 0\right\}.
\end{align}
This results in a single scalar uncertainty, in contrast to parameterizing $\tilde{d}_k$ for each time step and restricting them to be equal as in \cite{Berberich2020}. Furthermore, we assume known bounds $\sigma_\mathrm{max}(\tilde{V}_X)^2\leq \bar{v}_x^2 (N-1)$, $\sigma_\mathrm{max}(\tilde{V}_{X_+})^2\leq \bar{v}_x^2(N-1)$  and $\sigma_\mathrm{max}(\tilde{V}_{Z_p})^2\leq \bar{v}_{Z_p}^2(N-1)$ with $\bar{v}_x=\bar{v}_{Z_p}=5\cdot 10^{-4}$.
This approximately achieves a constant signal-to-noise ratio \cite[Section II.C]{Berberich2020} and satisfaction of Assumption\,\ref{ass:Param:G}.
Note that each \ac{QMI} can be multiplied by a positive scalar without changing the corresponding set.
These additional multipliers are included as optimization variables to reduce conservatism.
The chosen error description results in an \ac{SDP}, whose size and number of decision variables are independent of the amount of samples $N$. 
We repeat each experiment $1000$ times by uniformly sampling the initial condition from $[-1,1]^4$ and $w_{p,k}\in [-1,1]^2$.
The unknown signals $\tilde{V}_X$, $\tilde{V}_{X_+}$, $\tilde{V}_{Z_p}$ and $\tilde{D}$ are randomly sampled within the described sets.
In the following, we compare two different choices for the right inverse $\hat{G}$ as in \eqref{eq:Analysis:Gopt} and the Moore-Penrose-inverse $G^\dagger$.

In Fig.\,\ref{fig:Example:Statistic}, we illustrate the relative error $\epsilon_G = \frac{|\gamma-\gamma_{\mathrm{tr}}|}{\gamma_{\mathrm{tr}}}$ of the upper bound with respect to the $\mathcal{H}_2$-norm $\gamma_{\mathrm{tr}}$ of the true system and the success rate $\alpha$ of solving the \ac{SDP}. 
Furthermore, we compare our approach to bounding the term $W\coloneq [A_{\mathrm{tr}}^\top,\:C_{\mathrm{p,tr}}^\top]^\top \tilde{L}_{V_X}\tilde{V}_X$ by $\sigma_\mathrm{max}(W)\leq \sigma_\mathrm{max} ([A_{\mathrm{tr}}^\top,\:C_{\mathrm{p,tr}}^\top]^\top \tilde{L}_{V_X}) \sigma_\mathrm{max}(\tilde{V}_X)$ to treat the measurement error as a disturbance similar to \cite{DePersis2020}.
We denote the results with the index $0$.
For small $N$, the upper bound is significantly larger than $\gamma_{\mathrm{tr}}$. 
For increasing $N$, the difference shrinks to about $20\%$ of the actual value.
This is the expected behavior, as an increasing number of data points also shrinks the size of the set $\Sigma_{\Theta}$ and its overapproximation.
The remaining error is due to the applied convex relaxations. 
Furthermore, in \cite[Corollary 3]{Braendle2024} it is shown, that even without error-in-variables, $\Sigma_{\Theta}$ only converges to the true system, if the error bound is tight and additionally the disturbance realization is of a particular structure.
Moreover, we observe that transforming the measurement error to a disturbance yields a larger upper bound.
This is expected, since this parameterization does not account for the particular structure of the measurement error.
Note that, there is  also a difference between $\hat{G}$ and $G^\dagger$. 
For sufficiently large $N$ $\hat{G}$ appears to give better results. 
This is expected due to $\hat{G}$ being chosen to minimize the size of the uncertainty set, hence leading to less conservatism.
For small $N$ this difference is not that conclusive.
During simulations small $N$  lead to different uncertainty realizations with \acp{LFT} close to the stability boundary, which causes large $\gamma$.
Increasing the number of experiments also did not yield conclusive evidence, which right inverse should be preferred.
Hence, it's unclear whether it's better to choose $G$ to minimize $YG$ to be farther away from the stability boundary or to reduce the size of the uncertainty set.

	\begin{figure}[t]\centering
		\includegraphics[width = 1.00\linewidth]{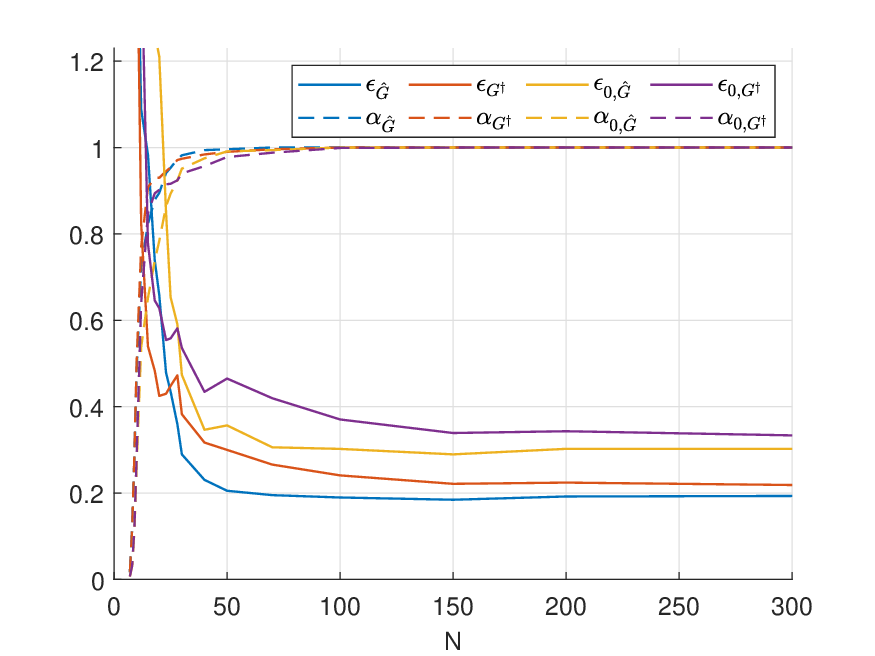}
		\caption{Error ratio of the robust $\mathcal{H}_2$-norm and the ratio of feasible \ac{SDP} in Theorem\,\ref{theorem:Analysis:RobustH2}, depending on the data length $N$, averaged over $1000$ experiments for different right inverses.}
		\label{fig:Example:Statistic}
	\end{figure}

At last, we analyze the feasibility of the \ac{SDP}. 
Fig.\,\ref{fig:Example:Statistic} shows the ratio of how often the \ac{SDP} is feasible compared to the total number of experiments.
This leads to similar results as before. 
For small sample sizes, the problem is rarely feasible, even though the \ac{LFT} is well-posed. 
This is due to the $\mathcal{H}_2$-norm requiring asymptotic stability of the system.
It becomes feasible with increasing $N$ such that an upper bound can be computed.
Similarly to before, this may be explained by the trade-off between small $YG$ and small uncertainty set.
However, a detailed theoretical or empirical case study is beyond the scope of this paper and is left for future research.

\section{CONCLUSION}
We have introduced a new \ac{LFT} parametrization for the linear regression problem with error-in-variables, exploiting the Sherman-Morrison-Woodbury formula.
Moreover, we demonstrated how to employ this data-driven parametrization to derive a guaranteed upper bound on the $\mathcal{H}_2$-norm of an unknown linear system using a finite number of noisy input-state measurements by solving an \ac{SDP}.
The presented framework allows to neatly incorporate many different setups, as designing controllers or input-output settings.
Future research should extend this framework, for example, by reducing the size of the uncertainty introduced by measurement errors and disturbances.

\bibliographystyle{IEEEtran}
\bibliography{Literature}

\end{document}